\documentclass[11pt,letterpaper,draftcls,oneside,peerreview,onecolumn]{IEEEtran}
\usepackage{amsfonts}
\usepackage{amssymb}
\usepackage{amsmath}
\usepackage[dvips]{graphicx}
\usepackage{cite}
\usepackage{balance}
\usepackage{caption}

\newtheorem{theorem}{Theorem}

\newtheorem{definition}{Definition}

\newtheorem{remark}{Remark}

\begin{document}

\title{Error Performance of Multidimensional Lattice Constellations - Part I: A Parallelotope Geometry Based Approach for the AWGN Channel}

\author{Koralia N. Pappi, \IEEEmembership{Student Member, IEEE,}  Nestor D. Chatzidiamantis, \IEEEmembership{Student Member, IEEE,} and George~K.~Karagiannidis, \IEEEmembership{Senior Member, IEEE}

\thanks{The authors are with the Electrical and Computer Engineering
        Department, Aristotle University of Thessaloniki, GR-54124 Thessaloniki,
        Greece (e-mails: \{kpappi, nestoras,
        geokarag\}@auth.gr).}
\thanks{This work has been presented in part at the 2012 IEEE Wireless Communications and Networking Conference (WCNC), Paris, France, April 2012.}} \maketitle

\begin{abstract}
Multidimensional lattice constellations which present signal space
diversity (SSD) have been extensively studied for single-antenna
transmission over fading channels, with focus on their optimal
design for achieving high diversity gain. In this two-part series
of papers we present a novel combinatorial geometrical approach
based on parallelotope geometry, for the performance evaluation of
multidimensional finite lattice constellations with arbitrary
structure, dimension and rank. In Part I, we present an analytical
expression for the exact symbol error probability (SEP) of
multidimensional signal sets, and two novel closed-form bounds,
named Multiple Sphere Lower Bound (MLSB) and Multiple Sphere Upper
Bound (MSUB). Part II extends the analysis to the transmission
over fading channels, where multidimensional signal sets are
commonly used to combat fading degradation. Numerical and
simulation results show that the proposed geometrical approach
leads to accurate and tight expressions, which can be efficiently
used for the performance evaluation and the design of
multidimensional lattice constellations, both in Additive White
Gaussian Noise (AWGN) and fading channels.
\end{abstract}


\newpage

\section{Introduction}\label{Intro}

The employment of Signal Space Diversity (SSD)-a method which has
been introduced in \cite{Boutros1} to compensate for the
degradation caused by fading channels-to multidimensional lattice
constellations, has attracted the interest of both academia and
industry. By performing component interleaving, new
multidimensional signal sets can be designed, which can achieve
diversity gain without any additional requirements for power,
bandwidth or multiple antennas, but only through rotation of the
multidimensional constellation. Such signal sets that have the
potential to achieve full diversity, have been presented in the
pioneer works \cite{Boutros1,Boutros2,Giraud,Viterbo2,Viterbo3}
and are carved from rotated multidimensional lattices, which meet
the criterion of the maximization of the minimum product distance.
Multidimensional constellations are also used in Multiple
Input-Multiple Output (MIMO) systems \cite{Lee, Fatma},
cooperative communication systems \cite{motahari} and various
coded schemes \cite{Yaboli,LeeLee, Tran2}, while SSD has been
included in the Second Generation Digital Terrestrial Television
Broadcasting System (DVB-T2) standard \cite{DVB-T2}.

\subsection{Motivation}\label{motivation}

Although the evaluation of the performance of such rotated
multidimensional signal sets can be an important tool in their
design, the study of the symbol error probability (SEP) is in
general a hard problem, both in Additive White Gaussian Noise
(AWGN) and in fading channels. This is mainly due to the
difficulty in the analytical computation of the Voronoi cells of
multidimensional constellations\cite{Viterbo}, and the fact that
fading acts independently upon each of the coordinates of the
signal, thus making stochastic not just the power but also the
structure of the lattice.

Various methods have been presented in order to evaluate the
performance of such signal sets, based on either approximations
\cite{Belfiore}, union bounds \cite{Taricco}, or bounds on the
maximization of the minimum product distance concerning algebraic
constructions, such as in \cite{Bayer}. Only recently, some exact
expressions for the SEP of two-dimensional constellations have
been presented in \cite{Jihoon} for Ricean fading channels;
however, the extension of such an analysis to multiple dimensions
seems to be complicated.

The sphere lower bound (SLB), which dates back to Shannon's work
\cite{Shannon}, has been proposed as an efficient tool for
evaluating the performance of multidimensional constellations. By
approximating the decision regions of infinite lattice
constellations - that is multidimensional constellations with
infinite number of points - with a sphere of the same volume, a
tight lower bound on their error performance can be obtained. This
bound in the presence of AWGN has been investigated in
\cite{Viterbo,Tarokh}, while in a similar manner, a sphere upper
bound (SUB) based on the packing radius of the lattice, has been
presented in \cite{Viterbo}. Although both of these sphere bounds
have been investigated in AWGN, their performance in the presence
of fading has not been thoroughly explored so far. In
\cite{Vialle}, the performance of SLB in Rayleigh channels was
approximated via a geometrical approach, while in \cite{Fabregas1}
it was evaluated for Nakagami-\emph{m} block fading channels
through numerical methods. However, it was clearly demonstrated
that, although it is a lower bound for infinite lattice
constellations, it is not generally a lower bound for finite
lattice constellations. Regarding the SUB, to the best of the
authors' knowledge, its performance in the presence of fading has
not been previously investigated. Moreover, while the SUB is an
upper bound also for finite lattice constellations, it is rather
loose.

\subsection{Contribution}\label{Contribution}

In this two-part paper, we provide an analytical framework for the
SEP evaluation of multidimensional finite lattice constellations.
Our analysis can be efficiently applied to multidimensional signal
sets, with arbitrary lattice structure, dimension and rank, taking
into account their common geometrical property: the constellations
form parallelotopes in the multidimensional signal space.

More specifically, in Part I we introduce a combinatorial approach
for the evaluation of the error performance of these signal sets,
based on the parallelotope geometry. Following this approach, we
derive an analytical expression for the exact SEP of
multidimensional finite lattice constellations, which is then
lower- and upper-bounded by two novel closed-form expressions,
called Multiple Sphere Lower Bound (MSLB) and Multiple Sphere
Upper Bound (MSUB) respectively. The MSLB is a new lower bound
which - in contrast with the SLB - takes into account the boundary
effects of a finite constellation. Similarly the MSUB, also taking
into account the boundary effects, is a tighter upper bound in
comparison with the SUB.

These expressions can be easily extended to multidimensional
signal sets distorted by fading. The error performance evaluation
in fading channels is investigated in Part II\cite{Pappi}.
Analytical expressions, which bound the frame error probability in
block fading channels, are derived for the MSLB and the MSUB,
while closed-form expressions are further presented for the SLB
and SUB in block fading. This set of expressions proves to be a
powerful tool for the error performance analysis of
multidimensional constellations, which employ SSD in order to
combat the fading degradation.

The remainder of the Part I is organized as follows. In Section
\ref{Lattices}, the structure and properties of infinite and
finite lattice constellations are described and the geometry of
multidimensional parallelotopes is discussed. Section
\ref{awgnper} presents the system model, while an expression for
the exact performance of finite lattice constellations in the AWGN
channel is derived and the the MSLB and MSUB are introduced. The
simulation results of various constellations and the analytical
bounds are discussed in Section \ref{results}.

\section{Lattices and Parallelotope Geometry\label{Lattices}}
\subsection{Infinite Lattice Constellations}\label{infiniteL}

An infinite lattice constellation lying in an $N$-dimensional
space consists of all the points of a lattice denoted by
$\Lambda.$ A lattice $\Lambda$ is called a \emph{full rank
lattice} when all of its points can be expressed in terms of a set
of $N$ independent vectors $\mathbf{v_i}$, $i=1,\ldots,N$, called
\emph{basis vectors}. In full rank lattices, every lattice point
is given by

\begin{equation}\label{Lambda}
\Lambda=\mathbf{M}\mathbf{z},\,\,\,\,\,\,\,\,\mathbf{z}\in\mathbb{Z}^N,
\end{equation}
where $\mathbf{M}\in\mathbb{R}^{N\times N}$ is the \emph{generator
matrix} and $\mathbf{z}\in\mathbb{Z}^N$ is a vector whose elements
are integers. Each different vector $\mathbf{z}$ corresponds to a
different point on the lattice $\Lambda$.

The columns of the generator matrix $\mathbf{M}$ are the basis
vectors $\mathbf{v_i}$, that is

\begin{equation}\label{generatorM}
\begin{array}{ccc}
\mathbf{M}=[\mathbf{v_1}\,\, \mathbf{v_2}\,\, \ldots\,\,
\mathbf{v_N}],&\mathbf{v_i}=[v_{i1}\,\,v_{i2}\,\,\ldots\,\,v_{iN}]^T,&i=1,2,\ldots,N.
\end{array}
\end{equation}
The parallelotope consisting of the points
\begin{equation}\label{parallelotope}
\theta_1\mathbf{v_1}+\theta_2\mathbf{v_2}+\ldots+\theta_N\mathbf{v_N},\,\,\theta_i=\{0,1\},
\end{equation}
is called the \emph{fundamental parallelotope} of the lattice
which tessellates Euclidean space. The volume of the fundamental
parallelotope is
$\mathrm{vol}(\Lambda)=|\mathrm{det}(\mathbf{M})|$.

We call the \emph{Voronoi cell}, $\mathcal{V}_\Lambda$, of a
lattice point $s_i$, the region $\mathbf{R}$ for which holds that
\cite{Viterbo}
\begin{equation}\label{voronoi}
\mathcal{V}_\Lambda=\{x\in \mathbf{R}:
\|x-s_i\|\leq\|x-s_j\|\text{ for all }i\neq j\}.
\end{equation}
In an infinite lattice constellation, the Voronoi cell also
tessellates Euclidean space, and thus, it is also
$\mathrm{vol}(\mathcal{V}_{\Lambda})=|\mathrm{det}(\mathbf{M})|$.
Next, this volume is normalized to be
$|\mathrm{det}(\mathbf{M})|=1$, as in \cite{Fabregas1,Conway}.

\subsection{Finite Lattice Constellations}\label{finiteL}

We consider finite lattice constellations, denoted by $\Lambda'$,
which are carved from an infinite $N$-dimensional lattice
constellation $\Lambda$ and they can be defined with respect to
the generator matrix $\mathbf{M}$ of the lattice $\Lambda$, from
which $\Lambda'$ is carved. Each of these constellations have $K$
points along the direction of each basis vector, thus having a
parallelotope as a shaping region, formed by the vector basis of
the infinite lattice constellation $\Lambda$. These constellations
will be denoted by a $K$-Pulse Amplitude Modulation ($K$-PAM),
since we assume that they are constructed by a PAM signal set
along each basis vector direction. Note that this is not the usual
consideration of multidimensional signal sets produced by a PAM
along every coordinate, since the basis vectors are not orthogonal
in the general case. A finite lattice constellation is defined as

\begin{equation}\label{LambdaF}
\begin{array}{ccc}
\Lambda'=\mathbf{M}\mathbf{u},&\mathbf{u}=[u_1\,\,u_2\,\,\ldots\,\,u_N]^T,&u_i\in\{0,1,\ldots,K-1\}.
\end{array}
\end{equation}
When a finite lattice is considered as a signal set, it is usually
in the form

\begin{equation}\label{Lambdaconst}
\Lambda'=\mathbf{M}\mathbf{u}+\mathbf{x_0},
\end{equation}
where $\mathbf{x_0}$ is an offset vector, used to minimize the
mean energy of the constellation. Since this does not affect our
analysis, it is omitted hereafter.

\subsection{Parallelotope Geometry}\label{sublattices}

The finite lattice constellations under consideration form
$N$-dimensional parallelotopes in the $N$-dimensional signal
space, formed by the same basis vectors $\mathbf{v_i}$ as the
lattices they are carved from. Next, some basic definitions are
given, which demonstrate important geometrical characteristics of
the $N$-dimensional parallelotopes.

\begin{definition}\label{defsets}

We define all the \emph{basis vector subsets}, containing $k$ out
of $N$ basis vectors $\mathbf{v_i}$, $k\leq N$, as

\begin{equation}\label{vset}
\mathcal{S}_{k,p}\subseteq\mathcal{S}_{N}=
\{\mathbf{v_1},\mathbf{v_2},\ldots,\mathbf{v_N}\},
\end{equation}
where $p=1,2,\ldots,\binom{N}{k}$ is an index enumerating all
different subsets with $k$ out of $N$ basis vectors. When $k=0$ or
$k=N$, it is $p=1$ and therefore it is omitted. When $k=0$,
$\mathcal{S}_0$ is the empty set.
\end{definition}

\begin{definition}\label{deffacets} In a parallelotope, the vertices,
edges, faces etc., are called \emph{facets}.
Each facet which lies in a $k$-dimensional subspace, the span of a
$\mathcal{S}_{k,p}$ basis vector subset, is denoted by
$\mathcal{F}_{k,p}$. When $k=N$, $\mathcal{F}_N$ denotes the inner
space of the parallelotope and the index $p=1$ is omitted. When
$k=0$, each zero-dimensional facet $\mathcal{F}_0$ denotes one
vertex, and the index $p=1$ is also omitted. Edges are
one-dimensional facets, faces are two dimensional facets etc.
\end{definition}

According to Definition \ref{deffacets}, each facet includes all
points $\mathbf{x}$ in the $N$-dimensional space, which satisfy
\begin{equation}\label{eqforfacets}
\mathcal{F}_{k,p}=\{\mathbf{x}=\mathbf{Mr},\,\,\,
\mathbb{R}^N\ni\mathbf{r}=[r_1,r_2,\ldots,r_N]^T
:\left\{\begin{array}{ll}
0<r_i<K-1,&i:\mathbf{v_i}\in\mathcal{S}_{k,p}\\
r_i=\{0,K-1\},&i:\mathbf{v_i}\not\in\mathcal{S}_{k,p}\\
\end{array}\right.\},
\end{equation}
where $\mathbf{M}$ is the generator matrix with the basis vectors
$\mathbf{v_i}$ and $\mathbf{r}$ is an $N$-dimensional real vector.
On a specific $\mathcal{F}_{k,p}$ facet, the values of the $r_i$'s
for which $i:\mathbf{v_i}\not\in\mathcal{S}_{k,p}$ remain
constant.

\begin{definition}\label{defeqfacets}
We call \emph{equivalent facets} those facets lying in
$k$-dimensional subspaces defined by the same basis vector subset
$\mathcal{S}_{k,p}$.
\end{definition}

According to (\ref{eqforfacets}), the number of vectors
$\mathbf{v_i}\not\in\mathcal{S}_{k,p}$ is $N-k$ and there are two
possible values for the corresponding $r_i$ elements of the vector
$\mathbf{r}$. Consequently, there are $2^{N-k}$ different
combinations and thus $2^{N-k}$ equivalent $\mathcal{F}_{k,p}$
facets on the $N$-dimensional parallelotope, for specific $k$ and
$p$. Furthermore, since there are $\binom{N}{k}$ different values
for the index $p=1,\ldots,\binom{N}{k}$, the total number of
$k$-dimensional facets is

\begin{equation}\label{facets}
n_k=2^{N-k}\binom{N}{k},\,\,\,\,\,0\leq k\leq N.
\end{equation}
For example, a three-dimensional parallelotope, called
parallelepiped, consists of twelve edges, which in groups of four
are equivalent, that is four $\mathcal{F}_{1,p}$ facets for each
$p=1,2,3$. Accordingly, there are six faces, which in groups of
two are equivalent, that is two $\mathcal{F}_{2,p}$ facets for
each  $p=1,2,3$.

Let $r_i^{\mathcal{F}_{k,p}}$ be the elements $r_i$ of the vector
$\mathbf{r}$ in (\ref{eqforfacets}) for a specific
$\mathcal{F}_{k,p}$. Then,

\begin{definition}\label{adjacent}
For a $\mathcal{F}_{k,p}$ facet, all those facets
$\mathcal{F}_{q,t}$, for which
$\mathcal{S}_{k,p}\subset\mathcal{S}_{q,t}$ and
$r_i^{\mathcal{F}_{q,t}}=r_i^{\mathcal{F}_{k,p}}$ $\forall
i:\mathbf{v_i}\not\in\mathcal{S}_{q,t}$, will be called
\emph{adjacent} facets to $\mathcal{F}_{k,p}$.
\end{definition}

In other words, in an adjacent facet $\mathcal{F}_{q,t}$, when
$r_i=0$ or $r_i=K-1$, the corresponding $r_i$ in
$\mathcal{F}_{k,p}$ is of the same value. Since there are $N-q$
vectors $\mathbf{v_i}\not\in\mathcal{S}_{q,t}$ and $N-k$ vectors
$\mathbf{v_i}\not\in\mathcal{S}_{k,p}$, for specific $q$, $k<q\leq
N$, the number of adjacent $q$-dimensional facets is
$\binom{N-k}{N-q}=\binom{N-k}{q-k}$, which is also the number of
different $\mathcal{S}_{q,t}$ sets for which
$\mathcal{S}_{k,p}\subset\mathcal{S}_{q,t}$. Consequently, the
number of all adjacent facets of any dimension is
$\sum\limits_{q=k+1}^N\binom{N-k}{q-k}$. Note that, according to
the definition above, all facets $\mathcal{F}_{q,t}$ adjacent to a
facet $\mathcal{F}_{k,p}$ are of greater dimension than
$\mathcal{F}_{k,p}$.

\subsection{Lattice Constellation Points}\label{points}

The finite constellations considered in this paper construct
lattice parallelotopes. Each point in this lattice lies on a
specific $\mathcal{F}_{k,p}$ facet or in the inner space
$\mathcal{F}_N$ of the parallelotope.

\begin{definition}\label{pointclass}
A point of an $N$-dimensional lattice parallelotope is considered
an $\mathcal{F}_{k,p}$ - point when it lies on an
$\mathcal{F}_{k,p}$ facet, that is when
\begin{equation}\label{eqforpoints}
\mathbf{x}=\mathbf{Mu},\,\,\,
\mathbb{Z}^N\ni\mathbf{u}=[u_1,u_2,\ldots,u_N]^T
:\left\{\begin{array}{ll}
0<u_i<K-1,&i:\mathbf{v_i}\in\mathcal{S}_{k,p}\\
u_i=\{0,K-1\},&i:\mathbf{v_i}\not\in\mathcal{S}_{k,p}\\
\end{array}\right..
\end{equation}
\end{definition}

From Definition \ref{pointclass}, it can be easily deduced that
the number of points on a $\mathcal{F}_{k,p}$ facet is

\begin{equation}\label{numofpoints}
(K-2)^k,\,\,\,\,\,0\leq k\leq N,
\end{equation}
since there are $(K-2)$ different possible values for every $u_i$
with $i:\mathbf{v_i}\in\mathcal{S}_{k,p}$, and there are $k$ such
values of $i$.

\begin{definition}\label{innerpoints}
All points for which $u_i\neq 0$ and $u_i\neq K-1$ $\forall i$ in
(\ref{eqforpoints}), are called \emph{inner points} of the
constellation. All the remaining points are called \emph{outer
points}.
\end{definition}

\begin{definition}\label{defeqpoints}
Points on equivalent $\mathcal{F}_{k,p}$ facets are called
\emph{equivalent points}, when for each
$i:\mathbf{v_i}\in\mathcal{S}_{k,p}$, the corresponding $u_i$
value of the vector $\mathbf{u}$ in (\ref{eqforpoints}), is equal
between all points.

\end{definition}

For example, in Fig. \ref{Fig:2DL},
$\mathcal{S}_{1,1}=\{\mathbf{v_1}\}$ and
$\mathcal{S}_{1,2}=\{\mathbf{v_2}\}$. We can decern two
$\mathcal{F}_{1,1}$ edges parallel to $\mathbf{v_1}$, two
$\mathcal{F}_{1,2}$ edges parallel to $\mathbf{v_2}$ and four
vertices. There are four inner points lying in $\mathcal{F}_2$,
two points on each equivalent $\mathcal{F}_{1,1}$ and
$\mathcal{F}_{1,2}$ and four vertices in total. Points \emph{A}
and \emph{B} are equivalent points according to Definition
\ref{defeqpoints}, since it is $u_2=2$ for both and they lie on
equivalent $\mathcal{F}_{1,2}$ facets.

It must be noted here that the outer points of a finite lattice
lying on a $\mathcal{F}_{k,p}$ facet, can also be considered as
being points of a sublattice, defined by the basis vector subset
$\mathcal{S}_{k,p}$. Accordingly, we define the following Voronoi
cells:
\begin{definition}\label{voronoiskp}
The $k$-dimensional Voronoi cell of a sublattice, defined by a
vector subset $\mathcal{S}_{k,p}$, is denoted by
$\mathcal{V}_{\mathcal{S}_{k,p}}$. For $k=N$,
$\mathcal{V}_{\mathcal{S}_{N}}\equiv\mathcal{V}_\Lambda$.
\end{definition}

\section{Performance Evaluation in Additive White Gaussian Noise (AWGN)}\label{awgnper}

In practical communication schemes using lattice constellations,
the transmitted signal point belongs to a finite lattice
constellation, as described in Section \ref{finiteL}. Next, the
communication system model is presented and the geometry of these
signal sets is examined.

\subsection{System Model}\label{comsys}
We consider communication in an AWGN channel where the received
signal vector is
\begin{equation}\label{rsymbol}
\mathbf{y}=\mathbf{x}+\mathbf{w},
\end{equation}
with $\mathbf{y}\in\mathbb{R}^N$ being the received
$N$-dimensional real signal vector, $\mathbf{x}\in\mathbb{R}^N$ is
the transmitted $N$-dimensional real signal vector and
$\mathbf{w}\in\mathbb{R}^N$ is the $N$-dimensional noise vector
whose samples are zero-mean Gaussian independent random variables
with variance $\sigma^2$. We define the signal-to-noise ratio
(SNR) as $\rho=\frac{1}{\sigma^2}$. The transmitted signal vector
$\mathbf{x}$ is a signal point in an infinite lattice
constellation $\Lambda$ or a finite lattice constellation
$\Lambda'$.

The conditional probability of receiving $\mathbf{y}$ while
transmitting $\mathbf{x}$ is

\begin{equation}\label{condprob}
p(\mathbf{y}|\mathbf{x})=(2\pi\sigma^2)^{-\frac{N}{2}}\exp\left(-\frac{1}{2\sigma^2}\|\mathbf{y}-\mathbf{x}\|^2\right),
\end{equation}
and Maximum Likelihood (ML) detection is employed at the receiver.

\subsection{Analytical Expressions for the Symbol Error Probability (SEP)}\label{exact}

In an infinite lattice constellation $\Lambda$, all signal points
are considered equiprobable and they have exactly the same error
performance since their Voronoi cells are equal. Thus the SEP of
an infinite lattice constellation is \cite{Fabregas1}

\begin{equation}\label{pinf}
P_\infty(\rho)=1-\int_{\mathcal{V}_{\Lambda}}
p(\mathbf{z})\mathrm{d}\mathbf{z}.
\end{equation}
The evaluation of $P_\infty(\rho)$ is often a tedious task due to
the difficulty of the computation of $\mathcal{V}_{\Lambda}$
\cite{Viterbo}. However, it can be approximated or bounded by
closed-form expressions as in \cite{Fabregas1}. To the best of the
authors' knowledge, a similar expression to (\ref{pinf}) for
finite lattice constellations does not exist, since the decision
regions of the outer points of these constellations do not lie in
regions equal to $\mathcal{V}_{\Lambda}$, a fact often referred to
as boundary effect \cite{Fabregas1}.

The SEP of a finite lattice constellation is given by

\begin{equation}\label{ps}
P_{K-PAM}(\rho)=1-\frac{\sum\limits_{i=1}^{K^N}\left[\int_{\mathbf{R_i}}
p(\mathbf{z})\mathrm{d}\mathbf{z}\right]}{K^N},
\end{equation}
where $\mathbf{R_i}$, $i=1,\ldots,K^N$, are the regions of correct
decision of the constellation signal points and $p(\mathbf{z})$ is
the $N$-dimensional probability density function (pdf) of AWGN as
defined in (\ref{condprob}). The decision regions $\mathbf{R_i}$
of the inner points of the constellation are equal to the Voronoi
cell $\mathcal{V}_{\Lambda}$, while those of the outer points are
generally unknown. In order to circumvent this, we employ a
geometrical technique, so as to express the sum of integrals in
(\ref{ps}) in terms of integrals on integration regions that are
Voronoi cells of the sublattices defined by the vector subsets
$\mathcal{S}_{k,p}$.

To derive an analytical expression for (\ref{ps}), it is necessary
first to proceed to a partitioning of the $N$-dimensional space in
the following regions:
\begin{itemize}

\item
The inner space of the parallelotope,
$\mathcal{D}_{\mathcal{F}_N}\equiv\mathcal{F}_N$, as defined in
(\ref{eqforfacets}).
\item
All the disjoint regions, denoted by
$\mathcal{D}_{\mathcal{F}_{k,p}}$, which are the projections of a
facet $\mathcal{F}_{k,p}$ to the directions vertical to this
facet. These regions are defined as
\begin{equation}\label{eqforpartition}
\begin{array}{cc}
\mathcal{D}_{\mathcal{F}_{k,p}}=\{\mathbf{y}=\mathbf{x}+\overline{\mathbf{V}}\mathbf{a},\,\,\,
\mathbf{a}\in\mathbb{R}^{N-k}_{+},\,\,\,
\mathbf{x}\in\mathcal{F}_{k,p}\},&0\leq k<N,
\end{array}
\end{equation}
where $\mathbf{x}$ are the points on a facet $\mathcal{F}_{k,p}$
as defined in (\ref{eqforfacets}), $\mathbf{a}$ is a vector of
dimension $(N-k)\times 1$ with positive real elements and
$\overline{\mathbf{V}}$ is an $N\times (N-k)$ matrix. If
$k<(N-1)$, its columns are the vertical vectors on all
$\mathcal{F}_{N-1,t}$ facets, which are adjacent to
$\mathcal{F}_{k,p}$ according to Definition \ref{adjacent}, with
outward direction compared to the parallelotope. The number of
$\mathcal{F}_{N-1,t}$ adjacent facets is $\binom{N-k}{q-k}$ for
$q=N-1$, that is $\binom{N-k}{N-1-k}=\binom{N-k}{1}=N-k$. If
$k=N-1$, then $\overline{\mathbf{V}}$ is an $N\times 1$ vector,
vertical to the $\mathcal{F}_{N-1,p}$ facet itself, with outward
direction compared to the parallelotope.

\end{itemize}

For example, in Fig. \ref{Fig:2DL}, the four partitions
$\mathcal{D}_{\mathcal{F}_0}$ which are highlighted extend to
infinity. Each corresponding matrix $\overline{\mathbf{V}}$ is a
$2\times 2$ matrix containing the vectors
$\overline{\mathbf{v_1}}$ and $\overline{\mathbf{v_2}}$, or their
negatives, i.e. with opposite direction. Thus, an integral on the
sum of these partitions equals an integral on the projection of
one of the equivalent $\mathcal{F}_0$ facets to all directions
vertical to it.

\begin{remark}\label{decisionparts}
The outer points of a finite lattice constellation lie in decision
regions which extend to the infinity. Taking into account that
these regions are constructed by employing the ML criterion, for a
signal point lying on a $\mathcal{F}_{k,p}$ facet, the decision
region can be divided into partial regions. Each of them belongs
ether to the inner space $\mathcal{D}_{\mathcal{F}_{N}}$, the
region $\mathcal{D}_{\mathcal{F}_{k,p}}$ or the regions
$\mathcal{D}_{\mathcal{F}_{q,t}}$, where $\mathcal{F}_{q,t}$ is a
facet adjacent to $\mathcal{F}_{k,p}$, $q<N$. Consequently, for a
point lying on some $\mathcal{F}_{k,p}$ with decision region
$\mathbf{R}$ it holds that
\begin{equation}\label{decregion}
\int_{\mathbf{R}}
p(\mathbf{z})\mathrm{d}\mathbf{z}=\sum\limits_{i=k}^N\sum\limits_{j:\mathcal{S}_{k,p}\subseteq\mathcal{S}_{i,j}}\int_{D\in\mathcal{D}_{\mathcal{F}_{i,j}}}
p(\mathbf{z})\mathrm{d}\mathbf{z},
\end{equation}
where $D\in\mathcal{D}_{\mathcal{F}_{i,j}}$ is the part of the
decision region in the partition
$\mathcal{D}_{\mathcal{F}_{i,j}}$. The summation in
(\ref{decregion}) ensures that the facets considered are the facet
$\mathcal{F}_{k,p}$ on which the point lies and all of its
adjacent facets.
\end{remark}

For example, in Fig. \ref{Fig:2DL}, point A lies on a
$\mathcal{F}_{1,2}$ facet. According to Definition \ref{adjacent},
the only adjacent facet to $\mathcal{F}_{1,2}$, is the inner space
of the constellation $\mathcal{F}_2$. Thus, according to Remark
\ref{decisionparts}, the decision region of A is divided in two
parts, $D_{1A}$ and $D_{2A}$, with
$D_{1A}\in\mathcal{D}_{\mathcal{F}_2}$ and
$D_{2A}\in\mathcal{D}_{\mathcal{F}_{1,2}}$.

\begin{definition}\label{defJk}
An integral $J_{k,p}$ is defined as

\begin{equation}\label{jk}
J_{k,p}=\int_{\mathcal{V}_{\mathcal{S}_{k,p}}}
p(\mathbf{z_k})\mathrm{d}\mathbf{z_k},\,\,\,\,\,\,0<k<n,
\end{equation}
where $p(\mathbf{z_k})$ is a $k$-dimensional zero mean Gaussian
distribution, $\mathcal{V}_{\mathcal{S}_{k,p}}$ is the Voronoi
cell of the k-dimensional sublattice defined by the basis vector
subset $\mathcal{S}_{k,p}$. Note that when $k=0$, then
$J_0\triangleq1$.
\end{definition}

Let $\mathcal{L}_{k,p}$ be the number of equivalent
$\mathcal{F}_{k,p}$ facets for specific $k$ and $p$. If all the
integrals on the decision regions of $\mathcal{L}_{k,p}$
equivalent $\mathcal{F}_{k,p}$-points are added, the resulting sum
$\mathrm{S}$ is

\begin{equation}\label{sumeqpoints}
\mathrm{S}=\sum\limits_{\mathcal{L}_{k,p}}\int_{\mathbf{R}}
p(\mathbf{z})\mathrm{d}\mathbf{z}=\sum\limits_{\mathcal{L}_{k,p}}\sum\limits_{i=k}^N\sum\limits_{j:\mathcal{S}_{k,p}\subseteq\mathcal{S}_{i,j}}\int_{D\in\mathcal{D}_{\mathcal{F}_{i,j}}}
p(\mathbf{z})\mathrm{d}\mathbf{z},
\end{equation}
and since the decision regions
$D\in\mathcal{D}_{\mathcal{F}_{i,j}}$ are disjoint for different
points, (\ref{sumeqpoints}) yields

\begin{equation}\label{sumeqpoints2}
\mathrm{S}=\sum\limits_{i=k}^N\sum\limits_{j:\mathcal{S}_{k,p}\subseteq\mathcal{S}_{i,j}}\int_{\sum\limits_{\mathcal{L}_{i,j}}
D\in\mathcal{D}_{\mathcal{F}_{i,j}}}
p(\mathbf{z})\mathrm{d}\mathbf{z},
\end{equation}
where $\sum\limits_{\mathcal{L}_{i,j}}
D\in\mathcal{D}_{\mathcal{F}_{i,j}}$ is the sum of partial
decision regions of $\mathcal{L}_{k,p}$ equivalent points, on all
$\mathcal{L}_{i,j}$ equivalent $\mathcal{F}_{i,j}$ facets. This
sum of partial decision regions is a region which is the
projection of a $\mathcal{V}_{\mathcal{S}_{i,j}}$ Voronoi cell to
all directions vertical to the span of the $\mathcal{S}_{i,j}$ set
of vectors. To reduce the integrals' dimension, a change of
variable and a Jacobian transformation is used, as in
\cite{Tarokh}, and thus (\ref{sumeqpoints2}) yields

\begin{equation}\label{sumeqpoints3}
\mathrm{S}=\sum\limits_{i=k}^N\sum\limits_{j:\mathcal{S}_{k,p}\subseteq\mathcal{S}_{i,j}}\int_{\mathcal{V}_{\mathcal{S}_{k,p}}}
p(\mathbf{z_k})\mathrm{d}\mathbf{z_k}=\sum\limits_{i=k}^N\sum\limits_{j:\mathcal{S}_{k,p}\subseteq\mathcal{S}_{i,j}}J_{i,j}.
\end{equation}

For example, in Fig. \ref{Fig:2DL}, points A and B are equivalent
points on $\mathcal{F}_{1,2}$ facets. Their decision regions are
divided in the partial regions $D_{1A}$, $D_{2A}$, $D_{1B}$ and
$D_{2B}$. The integrals on these partial regions are combined into
two new integrals denoted with $J_2$ and $J_{1,2}$.

Employing the above method, we can now present the following
theorem:

\begin{theorem}\label{propsum}
The SEP of a multidimensional finite lattice constellation is
given by

\begin{equation}\label{psfinal}
P_{K-PAM}(\rho)=1-\frac{\sum\limits_{k=0}^N(K-1)^k\sum\limits_{p=1}^{\binom{N}{k}}J_{k,p}}{K^N}.
\end{equation}
\end{theorem}

\begin{proof}
Due to Definition \ref{adjacent}, Remark \ref{decisionparts} and
(\ref{sumeqpoints3}), the sum of partial regions of equivalent
points, lying on all equivalent $\mathcal{F}_{k,p}$'s, for
specific $k$ and $p$, yields the sum of integrals,

\begin{equation}\label{suma}
\mathrm{S}=\left\{\begin{array}{lc} \sum\limits_{i=k}^N
\sum\limits_{j:\mathcal{S}_{k,p}\subseteq
\mathcal{S}_{i,j}}J_{i,j},&k\neq0,\\
\sum\limits_{i=0}^N\sum\limits_{j=1}^{\binom{N}{i}}J_{i,j},&k=0.
\end{array}\right.
\end{equation}

From (\ref{numofpoints}) and (\ref{suma}), the sum of integrals of
the regions of all points, lying on $\mathcal{F}_{k,p}$ facets for
specific $k$ and $p$, is

\begin{equation}\label{sumb}
\begin{array}{lc}
(K-2)^k\sum\limits_{i=k}^N\sum\limits_{j:\mathcal{S}_{k,p}\subseteq\mathcal{S}_{i,j}}J_{i,j},&0<k<N,\\
\sum\limits_{i=0}^N\sum\limits_{j=1}^{\binom{N}{i}}J_{i,j},&k=0.
\end{array}
\end{equation}

Adding the above sums for all values of $p$ and $k$ we have

\begin{equation}\label{sums}
\sum_{k=1}^N\sum_{p=1}^{\binom{N}{k}}(K-2)^k\sum_{i=k}^N\sum_{j:\mathcal{S}_{k,p}\subseteq\mathcal{S}_{i,j}}J_{i,j}+\sum_{i=0}^N\sum_{j=1}^{\binom{N}{i}}J_{i,j}.
\end{equation}
By changing the order of summing for indexes $i$ and $k$ in the
first term of (\ref{sums}), and combining the sums for the
enumeration indexes $p$ and $j$, due to the possible subsets and
the times that each $J_{i,j}$ appears, (\ref{sums}) yields

\begin{equation}\label{sumsb}
\sum_{i=1}^N\sum_{k=1}^i\binom{i}{k}(K-2)^k\sum_{j=1}^{\binom{N}{i}}J_{i,j}+\sum_{i=0}^N\sum_{j=1}^{\binom{N}{i}}J_{i,j},
\end{equation}
which can be written as

\begin{equation}\label{sumsc}
\sum_{i=0}^N\left(\sum_{k=0}^i\binom{i}{k}(K-2)^k-1\right)\sum_{j=1}^{\binom{N}{i}}J_{i,j}+\sum_{i=0}^N\sum_{j=1}^{\binom{N}{i}}J_{i,j},
\end{equation}
or equivalently

\begin{equation}\label{sums2}
\sum_{i=0}^N\sum_{k=0}^i\binom{i}{k}(K-2)^k\sum_{j=1}^{\binom{N}{i}}J_{i,j}.
\end{equation}
Due to the binomial theorem, (\ref{sums2}) reduces to

\begin{equation}\label{sumfinal}
\sum\limits_{i=0}^N(K-1)^i\sum\limits_{j=1}^{\binom{N}{i}}J_{i,j}.
\end{equation}

Using (\ref{sumfinal}), (\ref{ps}) yields (\ref{psfinal}) and this
concludes the proof.
\end{proof}

The expression in (\ref{psfinal}) cannot be directly evaluated,
except for special cases, since the analytical evaluation of
$\mathcal{V}_{\mathcal{S}_{k,p}}$ is generally a hard problem
\cite{Viterbo}. However, for the important case of SQAM
constellations, since the Voronoi cells are square,
(\ref{psfinal}) reduces to the well known closed-form SEP for the
SQAM \cite{Proakis}. In the following we propose closed-form lower
and upper bounds to $P_{K-PAM}\left(\rho\right)$, called Multiple
Sphere Lower Bound (MSLB) and Multiple Sphere Upper Bound (MSUB),
respectively. In these bounds, the integrals on the decision
regions of the signal points are substituted by integrals on
spheres of various dimensions.

\subsection{Multiple Sphere Lower Bound (MSLB)}\label{lowerbounds}

For the readers' convenience, we first present the Sphere Lower
Bound (SLB) for infinite lattice constellations, presented also in
\cite{Fabregas1}.

The error probability, $P_{\infty}\left(\rho\right)$, of an
infinite lattice constellation $\Lambda$ is lower-bounded by

\begin{equation}\label{slb}
P_{slb}(\rho)=1-\int\limits_{\mathcal{B}_N}p(\mathbf{z})\mathrm{d}\mathbf{z},
\end{equation}
where $\mathcal{B}_N$ is an $N$-dimensional sphere of the same
volume as the Voronoi cell $\mathcal{V}_{\mathcal{S}_N}$. Due to
the normalization $|\mathrm{det}(\mathbf{M})|=1$, the sphere
$\mathcal{B}_N$ is of unitary volume. It holds that \cite{Conway}
\begin{equation}\label{multivol}
\mathrm{vol}(\mathcal{B}_N)=\frac{\pi^{\frac{N}{2}}R_N^N}{\Gamma\left(\frac{N}{2}+1\right)}=1,
\end{equation}
where $R_N$ is the radius of the $N$-dimensional sphere, and
$\Gamma(\cdot)$ is the Gamma Function defined by \cite[Eq.
(8.310)]{Gradshteyn}. The radius $R_N$ is given by

\begin{equation}\label{radiusinf}
R_N^2=\frac{1}{\pi}\Gamma\left(\frac{N}{2}+1\right)^{\frac{2}{N}}.
\end{equation}

Subsequently, by substituting (\ref{radiusinf}) in (\ref{slb}) and
taking into account (\ref{condprob}), we get

\begin{equation}
P_{slb}(\rho)=1-\int\limits_{\mathcal{B}_N}p(\mathbf{z})\mathrm{d}\mathbf{z}=1-\left[1-\frac{\Gamma\left(\frac{N}{2},\frac{R_N^2}{2}\rho\right)}{\Gamma\left(\frac{N}{2}\right)}\right]=\frac{\Gamma\left(\frac{N}{2},\frac{R_N^2}{2}\rho\right)}{\Gamma\left(\frac{N}{2}\right)},
\end{equation}
where $\Gamma(a,x)=\int_x^{+\infty}t^{a-1}e^{-t}\mathrm{d}t$ is
the upper incomplete Gamma function defined in \cite[Eq.
(8.350)]{Gradshteyn}.

\begin{definition}\label{Ik}
We define the integrals
\begin{equation}\label{Ik1}
I_k=\int\limits_{\mathcal{B}_k}
p(\mathbf{z_k})\mathrm{d}\mathbf{z_k},\,\,\,\,\,k=1,\ldots,N,
\end{equation}
where $\mathcal{B}_k$ is a $k$-dimensional sphere of radius $R_k$
and $p(\mathbf{z_k})$ is a $k$-dimensional zero mean Gaussian
distribution. When $k=0$, we define $I_0\triangleq J_0=1$.
\end{definition}

The above integrals can be written as \cite{Fabregas1}
\begin{equation}\label{Ik2}
I_k=\left\{\begin{array}{cc}
1,&k=0\\
1-\frac{\Gamma\left(\frac{k}{2},\frac{R_k^2}{2}\rho\right)}{\Gamma\left(\frac{k}{2}\right)},&k=1,2,\ldots,N\\
\end{array}\right.
\end{equation}
Similar to (\ref{radiusinf}), with a slight modification for
finite constellations, the radius $R_k$ in AWGN channels is
defined as follows.

\begin{definition}
The sphere radius $R_k$ is given by
\begin{equation}\label{radius}
R_k^2=\left\{\begin{array}{cc}
\frac{1}{\pi}\Gamma(\frac{k}{2}+1)^\frac{2}{k}W^2,&k=1,2,\ldots,(N-1)\\
\frac{1}{\pi}\Gamma(\frac{k}{2}+1)^\frac{2}{k},&k=N\\
\end{array}\right.
\end{equation}
where $W$ is

\begin{equation}\label{W}
W=\frac{\|\mathbf{v_1}\|+\|\mathbf{v_2}\|+\ldots+\|\mathbf{v_N}\|}{N},
\end{equation}
with $\|\mathbf{v_i}\|$ being the norm of basis vector
$\mathbf{v_i}$. Note that for $\mathbb{Z}^N$ lattices, $W=1$.
\end{definition}

\begin{theorem}\label{MSLB}

The SEP of a multidimensional finite lattice constellation is
lower bounded by

\begin{equation}\label{bound}
P_{mslb}(\rho)=1-\frac{\sum\limits_{k=0}^N{(K-1)^k\binom{N}{k}I_k}}{K^N},
\end{equation}
where $P_{mslb}(\rho)$ is called Multiple Sphere Lower Bound
(MSLB).
\end{theorem}

\begin{proof}
The volume of $\mathcal{V}_{\mathcal{S}_{k,p}}$ in (\ref{jk}), is
the volume of Voronoi cell of a sublattice built by the basis
vector subset $\mathcal{S}_{k,p}$. Since this volume is the same
as the volume of the corresponding fundamental parallelotope of
the sublattice, as a consequence of Hadamard's inequality, it
holds that

\begin{equation}\label{voronoivol}
\mathrm{vol}_k(\mathcal{V}_{\mathcal{S}_{k,p}})\leq\prod_{i:\mathbf{v}_i\in\mathcal{S}_{k,p}}\|\mathbf{v}_i\|,
\end{equation}
where the equality holds only when the vectors of
$\mathcal{S}_{k,p}$ are orthogonal and $\mathrm{vol}_k(\cdot)$ is
the $k$-dimensional volume of a region.

From (\ref{voronoivol}) it is

\begin{equation}\label{voronineq}
\sum\limits_{p=1}^{\binom{N}{k}}\mathrm{vol}_k(\mathcal{V}_{\mathcal{S}_{k,p}})\leq\sum\limits_{p=1}^{\binom{N}{k}}\prod_{i:\mathbf{v}_i\in\mathcal{S}_{k,p}}\|\mathbf{v}_i\|,
\end{equation}
which can be written as
\begin{equation}\label{voronineq2}
\sum\limits_{p=1}^{\binom{N}{k}}\mathrm{vol}_k(\mathcal{V}_{\mathcal{S}_{k,p}})\leq\sum\limits_{\substack{b_1+b_2+\ldots+b_N=k\\b_1,b_2,\ldots,b_N\in\{0,1\}}}\|\mathbf{v}_1\|^{b_1}\|\mathbf{v}_2\|^{b_2}\cdots\|\mathbf{v}_N\|^{b_N}.
\end{equation}

Using Maclaurin's Inequality \cite[p.52]{Hardy}, for
$a_1,a_2,\ldots,a_N\in\mathbb{R}$ and $0<k<N$,

\begin{equation}\label{Maclaurin}
\varrho_N^{\frac{1}{N}}\leq\varrho_k^{\frac{1}{k}}\leq\varrho_1,
\end{equation}
where
\begin{equation}\label{varrho}
\varrho_k=\frac{\sum\limits_{\substack{b_1+b_2+\ldots+b_N=k\\b_1,b_2,\ldots,b_N\in\{0,1\}}}a_1^{b_1}a_2^{b_2}\cdots
a_N^{b_N}}{\binom{N}{k}}.
\end{equation}

If we set $a_i=\|\mathbf{v}_i\|$, $i=1,2,\ldots,N$, then
$\varrho_1=W$ and from (\ref{Maclaurin}) and (\ref{varrho})

\begin{equation}\label{normineq}
\sum\limits_{\substack{b_1+b_2+\ldots+b_N=k\\b_1,b_2,\ldots,b_N\in\{0,1\}}}\|\mathbf{v}_1\|^{b_1}\|\mathbf{v}_2\|^{b_2}\cdots\|\mathbf{v}_N\|^{b_N}\leq\binom{N}{k}W^k.
\end{equation}

From (\ref{voronineq2}) and (\ref{normineq}), for $0<k<N$, we have

\begin{equation}\label{sphineq}
\sum\limits_{p=1}^{\binom{N}{k}}\mathrm{vol}_k(\mathcal{V}_{\mathcal{S}_{k,p}})\leq\binom{N}{k}W^k.
\end{equation}

Due to the spherical symmetry of the AWGN pdf, it is

\begin{equation}\label{regineq}
\int_Dp(\mathbf{z_k})\mathrm{d}(\mathbf{z_k})\leq\int_{\mathcal{B}_D}p(\mathbf{z_k})\mathrm{d}\mathbf{z_k},
\end{equation}
when $\mathrm{vol}_k(D)=\mathrm{vol}_k(\mathcal{B}_D)$, as in
\cite{Fabregas1}. In (\ref{regineq}) $D$ is a random
$k$-dimensional region of integration and $\mathcal{B}_D$ is a
$k$-dimensional sphere of the same volume. Thus, from (\ref{jk})
and (\ref{regineq}), it holds that

\begin{equation}\label{smallspheres}
J_{k,p}=\int_{\mathcal{V}_{\mathcal{S}_{k,p}}}
p(\mathbf{z_k})\mathrm{d}\mathbf{z_k}\leq
\int_{\mathcal{B}(\mathcal{S}_{k,p})}
p(\mathbf{z_k})\mathrm{d}\mathbf{z_k},
\end{equation}
where $\mathcal{B}(\mathcal{S}_{k,p})$ is a sphere with volume
$\mathrm{vol}_k\left(\mathcal{B}(\mathcal{S}_{k,p})\right)=\mathrm{vol}_k\left(\mathcal{V}_{\mathcal{S}_{k,p}}\right)$.
Subsequently,

\begin{equation}\label{sumsmallspheres}
\sum\limits_{p=1}^{\binom{N}{k}}J_{k,p}\leq\sum\limits_{p=1}^{\binom{N}{k}}
\int_{\mathcal{B}(\mathcal{S}_{k,p})}
p(\mathbf{z_k})\mathrm{d}\mathbf{z_k}=\sum\limits_{p=1}^{\binom{N}{k}}\left(1-\frac{\Gamma\left(\frac{k}{2},\frac{R_{\mathcal{S}_{k,p}}^2}{2}\rho\right)}{\Gamma\left(\frac{k}{2}\right)}\right),
\end{equation}
where $R_{\mathcal{S}_{k,p}}$ is the radius of the sphere
$\mathcal{B}(\mathcal{S}_{k,p})$. From (\ref{sphineq}), and using
that
$\mathrm{vol}_k\left(\mathcal{B}(\mathcal{S}_{k,p})\right)=\frac{\pi^{\frac{k}{2}}R_{\mathcal{S}_{k,p}}^k}{\Gamma\left(\frac{k}{2}+1\right)}$
as in (\ref{multivol}), it is

\begin{equation}\label{radiusineq1}
\sum\limits_{p=1}^{\binom{N}{k}}\frac{\pi^{\frac{k}{2}}R_{\mathcal{S}_{k,p}}^k}{\Gamma\left(\frac{k}{2}+1\right)}\leq\binom{N}{k}W^k,
\end{equation}
or by taking into account (\ref{radius}) for $0<k<N$,

\begin{equation}\label{radiusineq2}
\sum\limits_{m=1}^{\binom{N}{k}}R_{\mathcal{S}_{k,p}}^k\leq\binom{N}{k}R_k^k.
\end{equation}

Now, if $a,b$ are positive real numbers, the function
$f(x;a,b)=\Gamma \left( a,b x^{1/a} \right)$ is convex in
$(0,\infty)$. Indeed

\begin{equation}\label{gammad}
\frac{\partial{f}}{\partial{x}}= (b x^{1/a})^{a-1}e^{-b
x^{1/a}}\frac{\partial{(b x^{1/a})}}{\partial{x}}=-\frac{b^a e^{-b
x^{\frac{1}{a}}}}{a}
\end{equation}
and

\begin{equation}\label{gammadd}
\frac{\partial^2{f}}{\partial{x^2}}=\frac{b^{a+1}
x^{\frac{1}{a}-1} e^{-b x^{\frac{1}{a}}}}{a^2} > 0,\,\,\forall
x>0.
\end{equation}

Thus from Jensen's Inequality for convex functions \cite{Hardy}
\begin{equation}\label{jensen}
\sum_{i=1}^L {\Gamma \left( a,b {x_i}^{1/a} \right)} \geq L\Gamma
\left( a,b \left( \sum_{i=1}^L {x_i} /L \right) ^{1/a} \right).
\end{equation}

For $a=\frac{k}{2}$ , $b=\frac{\rho}{2}$, $L=\binom{N}{k}$ and
$x_i={R^k_{\mathcal{S}_{k,p}}}$ we get

\begin{equation}\label{gammaineq1}
\sum_{p=1}^{\binom{N}{k}} {\Gamma \left(
\frac{k}{2},\frac{\rho}{2} R^2_{\mathcal{S}_{k,p}} \right)} \geq
\binom{N}{k}\Gamma \left( \frac{k}{2},\frac{\rho}{2} \left(
\frac{\sum\limits_{m=1}^{\binom{N}{k}}
R^k_{\mathcal{S}_{k,p}}}{\binom{N}{k}} \right) ^{\frac{2}{k}}
\right).
\end{equation}

From (\ref{radiusineq2}) and since $f(x;a,b)=\Gamma \left( a,b
x^{1/a} \right)$ is a decreasing function

\begin{equation}\label{gammaineq2}
\Gamma \left( \frac{k}{2},\frac{\rho}{2} \left(
\frac{\sum\limits_{p=1}^{\binom{N}{k}}
R^k_{\mathcal{S}_{k,p}}}{\binom{N}{k}} \right) ^{\frac{2}{k}}
\right)\geq\Gamma \left( \frac{k}{2},\frac{\rho}{2}R^2_k \right).
\end{equation}

From (\ref{gammaineq1}) and (\ref{gammaineq2}), for $0<k<N$

\begin{equation}\label{gammaineq3}
\sum_{p=1}^{\binom{N}{k}} {\Gamma \left(
\frac{k}{2},\frac{\rho}{2} R^2_{\mathcal{S}_{k,p}} \right)}
\geq\binom{N}{k}\Gamma \left( \frac{k}{2},\frac{\rho}{2} R^2_k
\right),
\end{equation}
or equivalently

\begin{equation}\label{gammaineq4}
\sum_{p=1}^{\binom{N}{k}} \left(1-\frac{\Gamma \left(
\frac{k}{2},\frac{\rho}{2} R^2_{\mathcal{S}_{k,p}}
\right)}{\Gamma\left(\frac{k}{2}\right)}\right)
\leq\binom{N}{k}\left(1-\frac{\Gamma \left(
\frac{k}{2},\frac{\rho}{2} R^2_k
\right)}{\Gamma\left(\frac{k}{2}\right)}\right).
\end{equation}

Taking into account (\ref{sumsmallspheres}) and (\ref{gammaineq4})
for some $k$, $0<k<N$,  it yields

\begin{equation}\label{ineq_k}
\sum\limits_{p=1}^{\binom{N}{k}}J_{k,p}\leq\binom{N}{k}\left(1-\frac{\Gamma
\left( \frac{k}{2},\frac{\rho}{2} R^2_k
\right)}{\Gamma\left(\frac{k}{2}\right)}\right)=\binom{N}{k}I_k,
\end{equation}
while for $k=0$, $p=1$ and it holds that $J_0=I_0=1$.

For $k=N$, it is also $p=1$ and from (\ref{radius}) and
(\ref{smallspheres})

\begin{equation}\label{ineq_N}
J_N\leq\left(1-\frac{\Gamma \left( \frac{N}{2},\frac{\rho}{2}
R^2_N \right)}{\Gamma\left(\frac{N}{2}\right)}\right)=I_N.
\end{equation}

Combining (\ref{ineq_k})and (\ref{ineq_N}), multiplying by
$\left(K-1\right)^k$ and summing for all $k$, it yields

\begin{equation}\label{finalinequality}
\sum\limits_{k=0}^N\left(K-1\right)^k\sum\limits_{p=1}^{\binom{N}{k}}J_{k,p}\leq\sum\limits_{k=0}^N\left(K-1\right)^k\binom{N}{k}I_k.
\end{equation}

Using (\ref{psfinal}), (\ref{bound}) and (\ref{finalinequality}),

\begin{equation}\label{proofslb}
P_{mslb}(\rho)\leq P(\rho)
\end{equation}
and this concludes the proof.
\end{proof}

\subsection{Multiple Sphere Upper Bound (MSUB)}\label{upperbounds}

A well known upper bound for infinite lattice constellations,
which is based on the minimum distance between signal points, is
the Sphere Upper Bound (SUB) \cite{Viterbo}

\begin{equation}\label{sub}
P_{sub}(\rho)=1-\int\limits_{\mathcal{G}_N}p(\mathbf{z})\mathrm{d}\mathbf{z},
\end{equation}
where $\mathcal{G}_N$ is an $N$-dimensional sphere, with radius
defined by

\begin{equation}\label{radiussub}
\mathcal{R}^2=\left(\frac{d_{min}}{2}\right)^2=\frac{d_{min}^2}{4},
\end{equation}
with $d_{min}$ being the minimum distance on the infinite lattice
constellation $\Lambda$. That is, the sphere $\mathcal{G}_N$ is
inscribed in the Voronoi cell of the lattice.

When the generator matrix $\mathbf{M}$ is constructed by the basis
vectors $\mathbf{v_i}$, $i=1,2,\ldots,N$ of the minimum possible
norms, the minimum distance $d_{min}$ can be directly evaluated by
$d_{min}=\min_{i}\left\|\mathbf{v_i}\right\|$. Although this is
not always the case, the above is valid for the most commonly used
lattices in practical cases, such as the $\mathbb{Z}^N$ lattices.
Especially for the $\mathbb{Z}^N$ lattices, $d_{min}=1$.

The SUB in (\ref{sub}) can be rewritten as

\begin{equation}\label{SUBfinal}
P_{sub}(\rho)=1-\left[1-\frac{\Gamma\left(\frac{N}{2},\frac{\mathcal{R}^2}{2}\rho\right)}{\Gamma\left(\frac{N}{2}\right)}\right]=\frac{\Gamma\left(\frac{N}{2},\frac{\mathcal{R}^2}{2}\rho\right)}{\Gamma\left(\frac{N}{2}\right)}.
\end{equation}

Similarly, based on (\ref{psfinal}) and in the same concept as the
SUB for infinite lattice constellations, we can now provide a
novel upper bound for finite lattice constellations.

\begin{definition}\label{Isub}
We define the integrals
\begin{equation}\label{Iksub1}
\mathcal{I}_k=\int\limits_{\mathcal{G}_k}
p(\mathbf{z_k})\mathrm{d}\mathbf{z_k},\,\,\,\,\,k=0,1,\ldots,N,
\end{equation}
where $\mathcal{G}_k$ is a $k$-dimensional sphere, with radius
defined in (\ref{radiussub}). When $k=0$, we define
$\mathcal{I}_0=J_0=1$.
\end{definition}

The above integrals can be written as \cite{Fabregas1}
\begin{equation}\label{Iksub2}
\mathcal{I}_k=\left\{\begin{array}{cc}
1,&k=0\\
1-\frac{\Gamma\left(\frac{k}{2},\frac{\mathcal{R}^2}{2}\rho\right)}{\Gamma\left(\frac{k}{2}\right)},&k=1,2,\ldots,N.\\
\end{array}\right.
\end{equation}

\begin{theorem}\label{MSUB}
The SEP of a multidimensional finite lattice constellation is
upper bounded by

\begin{equation}\label{boundmsub}
P_{msub}(\rho)=1-\frac{\sum\limits_{k=0}^N{(K-1)^k\binom{N}{k}\mathcal{I}_k}}{K^N},
\end{equation}
where $P_{msub}(\rho)$ is called Multiple Sphere Upper Bound
(MSUB).
\end{theorem}

\begin{proof}
If $d_{min}(\mathcal{S}_{k,p})$ is the minimum distance between
signal points on the sublattice defined by the basis vector subset
$\mathcal{S}_{k,p}$, for any $J_{k,p}$, computed on a Voronoi cell
$\mathcal{V}_{\mathcal{S}_{k,p}}$

\begin{equation}\label{smallspheressub}
J_{k,p}=\int_{\mathcal{V}_{\mathcal{S}_{k,p}}}
p(\mathbf{z_k})\mathrm{d}\mathbf{z_k}\geq
\int_{\mathcal{G}(\mathcal{S}_{k,p})}
p(\mathbf{z_k})\mathrm{d}\mathbf{z_k},
\end{equation}
where $\mathcal{G}(\mathcal{S}_{k,p})$ is a $k$-dimensional sphere
with radius
$\mathcal{R}_{\mathcal{S}_{k,p}}=\frac{d_{min}(\mathcal{S}_{k,p})}{2}$.
The sphere $\mathcal{G}(\mathcal{S}_{k,p})$ is inscribed in the
Voronoi cell $\mathcal{V}_{\mathcal{S}_{k,p}}$. It is generally
valid that $d_{min}(\mathcal{S}_{k,p})\geq d_{min}$, where
$d_{min}$ is the minimum distance on the lattice defined by the
basis vector set $\mathcal{S}_N$. This is straightforward, since
$\mathcal{S}_{k,p}\subseteq\mathcal{S}_N$.

Thus,

\begin{equation}\label{sphereineq}
\int_{\mathcal{G}(\mathcal{S}_{k,p})}
p(\mathbf{z_k})\mathrm{d}\mathbf{z_k}\geq\int_{\mathcal{G}_k}
p(\mathbf{z_k})\mathrm{d}\mathbf{z_k},
\end{equation}
where $\mathcal{G}_k$ is a $k$-dimensional sphere with radius
$\mathcal{R}=\frac{d_{min}}{2}$, as defined in (\ref{radiussub}).
The sphere $\mathcal{G}_k$ is always smaller or at the most equal
to the inscribed sphere of the Voronoi cell
$\mathcal{V}_{\mathcal{S}_{k,p}}$.

Taking into account (\ref{Iksub1}), (\ref{smallspheressub}) and
(\ref{sphereineq}), it is $ J_{k,p}\geq\mathcal{I}_k$ and
subsequently,

\begin{equation}\label{finalinequalitysub}
\sum\limits_{k=0}^N\left(K-1\right)^k\sum\limits_{p=1}^{\binom{N}{k}}J_{k,p}\geq\sum\limits_{k=0}^N\left(K-1\right)^k\binom{N}{k}\mathcal{I}_k.
\end{equation}

From (\ref{psfinal}), (\ref{boundmsub}) and
(\ref{finalinequalitysub}),

\begin{equation}\label{proofsub}
P_{msub}(\rho)\geq P(\rho)
\end{equation}
and this concludes the proof.
\end{proof}

\section{Numerical Results \& Discussion}\label{results}

In this section we illustrate the accuracy and tightness of the
proposed lower and upper bounds, MSLB and MSUB, respectively, in
comparison with the SEP, as approximated by Monte-Carlo
simulation, for various finite lattice constellations in AWGN
channels. We also compare the MSLB and MSUB with the existing
bounds for the infinite lattice constellations, the SLB and SUB.
The lattice constellations most commonly used in practical cases
are those carved from $\mathbb{Z}^N$ lattices, due to the easy
Gray coded bit labeling. In the following, apart from
$\mathbb{Z}^N$ lattices, the $\mathbb{A}^2$, $\mathbb{E}^4$ and
$\mathbb{E}^8$ are also illustrated, as an example of lattice
structures different from the orthogonal constellations. These
schemes usually achieve better SEP but they cannot be labeled with
a Gray code.

Fig. \ref{Fig:z2-4pam} illustrates the performance of a
$\mathbb{Z}^2$ 4-PAM constellation, which is a simple case of
lattice constellations, most commonly named as 16-Square
Quadrature Amplitude Modulation (16-SQAM). The simulated SEP of
the constellation in the AWGN channel is plotted in conjunction
with the corresponding MSLB and MSUB for various values of the
SNR, $\rho=\frac{1}{\sigma^2}$. For the $\mathbb{Z}^N$ lattices,
the generator matrix is $\mathbf{M}=\mathbf{I}_N$, where
$\mathbf{I}_N$ is the $N\times N$ identity matrix, while
$W=d_{min}=1$. It is evident that the MSLB acts as a lower bound,
while the MSUB acts as an upper bound, for all values of $\rho$.
Both bounds are very tight and can be effectively used to assess
the performance of the $\mathbb{Z}^2$ 4-PAM constellation.
Compared to the existing SLB, the proposed MSLB corresponds better
to the actual performance of the constellation. Furthermore it is
evident that the SLB does not act as a lower bound for SNR values
lower than $15$dB, whereas the SLB becomes less tight than the
MSLB for SNR values higher than $17$dB. Finally, although the
existing SUB is an upper bound to the actual performance, the MSUB
is almost $0.5$dB tighter than the SUB.

Fig. \ref{Fig:z2-32pam} shows the performance of a $\mathbb{Z}^2$
32-PAM constellation. It is clearly illustrated that both the MSLB
and the MSUB bound the performance of the lattice and they are
still very tight, even if the rank of the $K$-PAM increases. In
this situation, the MSLB is almost in accordance with the SLB, and
the MSUB with the SUB respectively. This is because the inner
points are approximated in the same way and the ratio inner/outer
points on the constellation is higher than that of the 4-PAM
constellation. This implies that, for a specific dimension $N$, as
$K$ increases, the MSLB converges to the corresponding SLB, and
the MSUB converges to the SUB.

Figs. \ref{Fig:z4-4pam} and \ref{Fig:z8-4pam} depict the
performance of a $\mathbb{Z}^4$ 4-PAM and a $\mathbb{Z}^8$ 4-PAM
respectively, together with the corresponding MSLBs, MSUBs, SLBs
and SUBs. Comparing with Fig. \ref{Fig:z2-4pam}, where a
$\mathbb{Z}^2$ 4-PAM is illustrated, it is evident that, for a
specific $K$, as the dimension decreases, the bounds become more
tight. Still, for both dimensions, the proposed bounds are tighter
than the existing SLBs and SUBs, while for the SLB we can also see
that for low SNR values, it does not act as a bound. Moreover,
since MSLB and SLB diverge from each other for high SNR values,
the results also suggest that the MSLB has different diversity
order than the SLB, corresponding better to the diversity order of
the actual performance of the constellations.

In the following figures, the performance of some non orthogonal
lattices is depicted, in order to highlight the efficiency of the
MSLB and MSUB for various lattice structures. In Fig.
\ref{Fig:a2-4pam}, a $\mathbb{A}^2$ 4-PAM is illustrated. The
generator matrix is given by \cite{Conway}

\begin{equation}
\mathbf{M} = \left[ \begin{array}{cc}
\sqrt{\frac{2}{\sqrt{3}}}&\sqrt{\frac{1}{2\sqrt{3}}}\\
0&\sqrt{\frac{3}{2\sqrt{3}}} \end{array}\right],
\end{equation}
and thus $W=\sqrt{\frac{2}{\sqrt{3}}}$ and
$d_{min}=\sqrt{\frac{2}{\sqrt{3}}}$. Once again it is clear that
both the MSLB and MSUB are reliable and tight, in constrast to the
SLB and SUB. Specifically, the corresponding SLB is not a lower
bound for this case, for all SNR values considered. Moreover, the
proposed bounds are more tight than the case of $\mathbb{Z}^N$
lattices. This can be attributed to the structure of the
$\mathbb{A}^2$ lattice, since the Voronoi cells of these lattices
are regular polytopes, which are better approximated by the
spheres, used both in MSLB and MSUB.

In Fig. \ref{Fig:a2-32pam}, the rank $K$ of the $\mathbb{A}^2$
lattice is increased from $K=4$ to $K=32$.  Again, as $K$
increases, MSLB and MSUB converge to the corresponding SLB and
SUB, maintaining their accuracy and tightness.

In Figs. \ref{Fig:e4-4pam} and \ref{Fig:e8-4pam}, the lattices
$\mathbb{E}^4$ 4-PAM and $\mathbb{E}^8$ 4-PAM are presented
\cite{Conway,Thompson}. The generator matrices are given in
(\ref{E4E8}), while $W=d_{min}=\frac{2}{8^\frac{1}{4}}$ for $N=4$,
and $W=\frac{2+7\sqrt{2}}{8}$ and $d_{min}=\sqrt{2}$ for $N=8$.
Both MSLB and MSUB act as tight bounds, in contrast to the
corresponding SLB and SUB, while they are tighter than the
corresponding cases of the $\mathbb{Z}^N$ lattices.

\begin{equation}\label{E4E8}
\mathbf{M}_{\mathbb{E}^4} =\frac{1}{8^\frac{1}{4}} \left[
\begin{array}{cccc} 1&2&0&0\\1&0&2&0\\1&0&0&2\\1&0&0&0
\end{array}\right],\,\,\,\,\,\,\,\,\, \mathbf{M}_{\mathbb{E}^8}
=\left[ \begin{array} {cccccccc}
2&-1&0&0&0&0&0&1/2\\
0&1&-1&0&0&0&0&1/2\\
0&0&1&-1&0&0&0&1/2\\
0&0&0&1&-1&0&0&1/2\\
0&0&0&0&1&-1&0&1/2\\
0&0&0&0&0&1&-1&1/2\\
0&0&0&0&0&0&1&1/2
\\0&0&0&0&0&0&0&1/2 \end{array}\right].
\end{equation}

\section{Conclusions}\label{conclusions}

We studied the error performance of finite lattice constellations
via a combinatorial geometrical approach. First we presented an
analytical expression for the exact SEP of these signal sets,
which is then used to introduce two novel closed-form bounds,
called Multiple Sphere Lower Bound (MSLB) and Multiple Sphere
Upper Bound (MSUB). The accuracy and tightness of MSLB and MSUB
have been illustrated in comparison with the simulated SEP of
various constellations of different lattice structure, dimension
and rank. The proposed bounds are tighter to the actual
performance, compared to the SLB and SUB which are often used as
approximations for the finite case. The presented approach can be
extended to multidimensional signal sets distorted by fading, as
presented in Part II. Since these constellations illustrate
substantial diversity gains, the proposed analytical framework and
its extension to fading channels becomes an important and
efficient tool for their design and performance evaluation.

\newpage

\begin{figure}[h!]
\centering\includegraphics[keepaspectratio,width=5in]{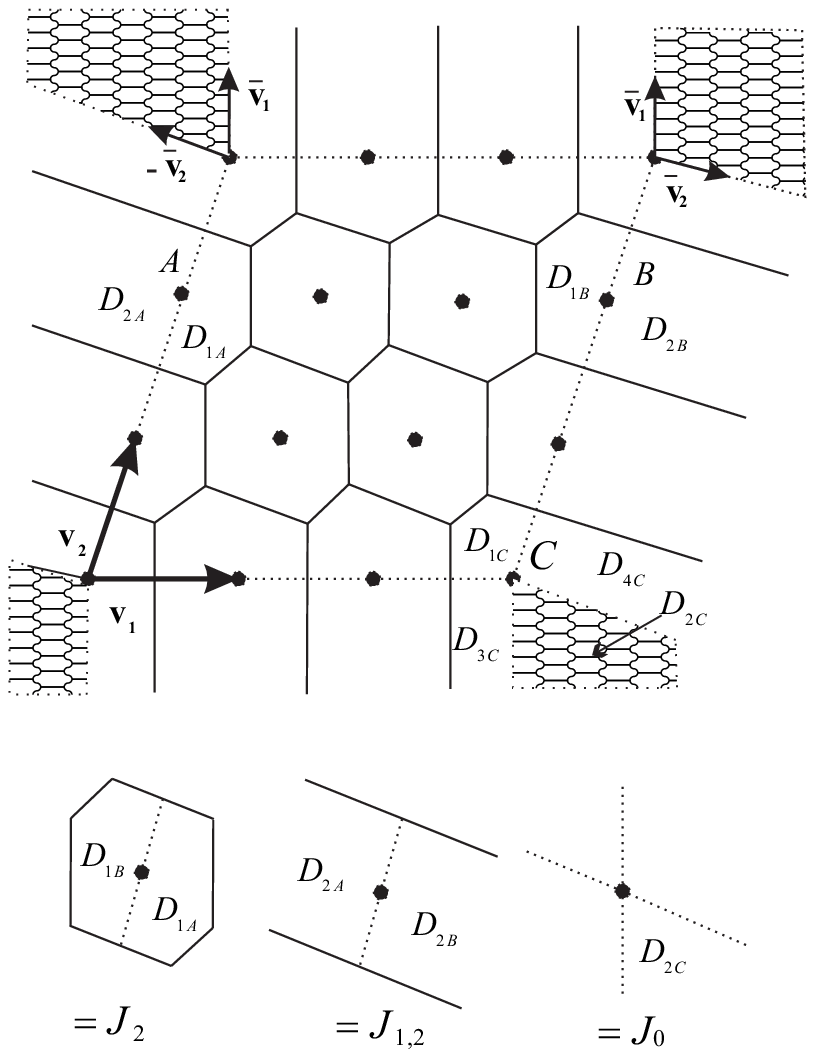}
\center\caption{2D Lattice and Decision Region Combining}
\label{Fig:2DL}
\end{figure}

\begin{figure}[h!]
\centering\includegraphics[keepaspectratio,width=6in]{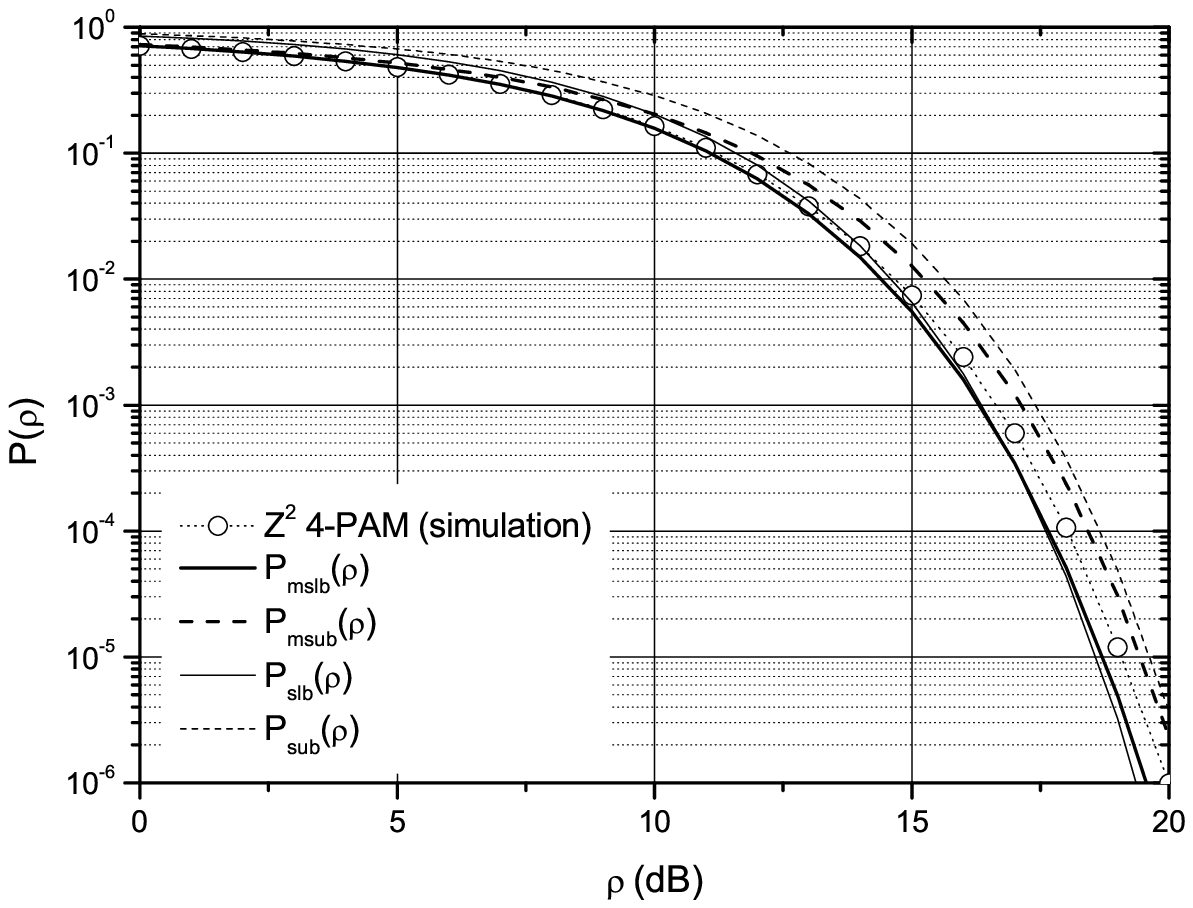}
\center\caption{Symbol Error Probability, MSLB and MSUB for the
$\mathbb{Z}^2$ $4-PAM$ constellation and SLB and SUB for the
$\mathbb{Z}^2$ lattice.} \label{Fig:z2-4pam}
\end{figure}

\begin{figure}[h!]
\centering\includegraphics[keepaspectratio,width=6in]{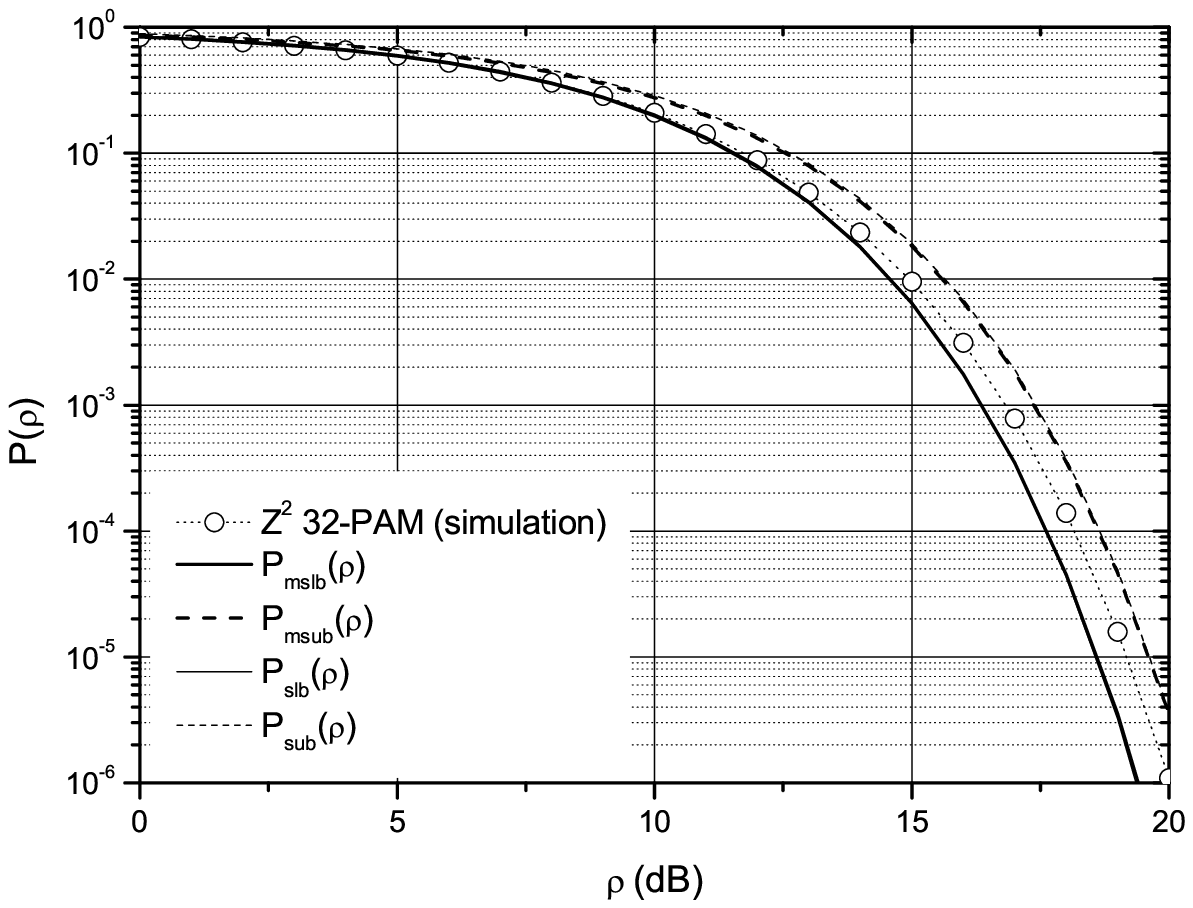}
\center\caption{Symbol Error Probability, MSLB and MSUB for the
$\mathbb{Z}^2$ $32-PAM$ constellation and SLB and SUB for the
$\mathbb{Z}^2$ lattice.} \label{Fig:z2-32pam}
\end{figure}

\begin{figure}[h!]
\centering\includegraphics[keepaspectratio,width=6in]{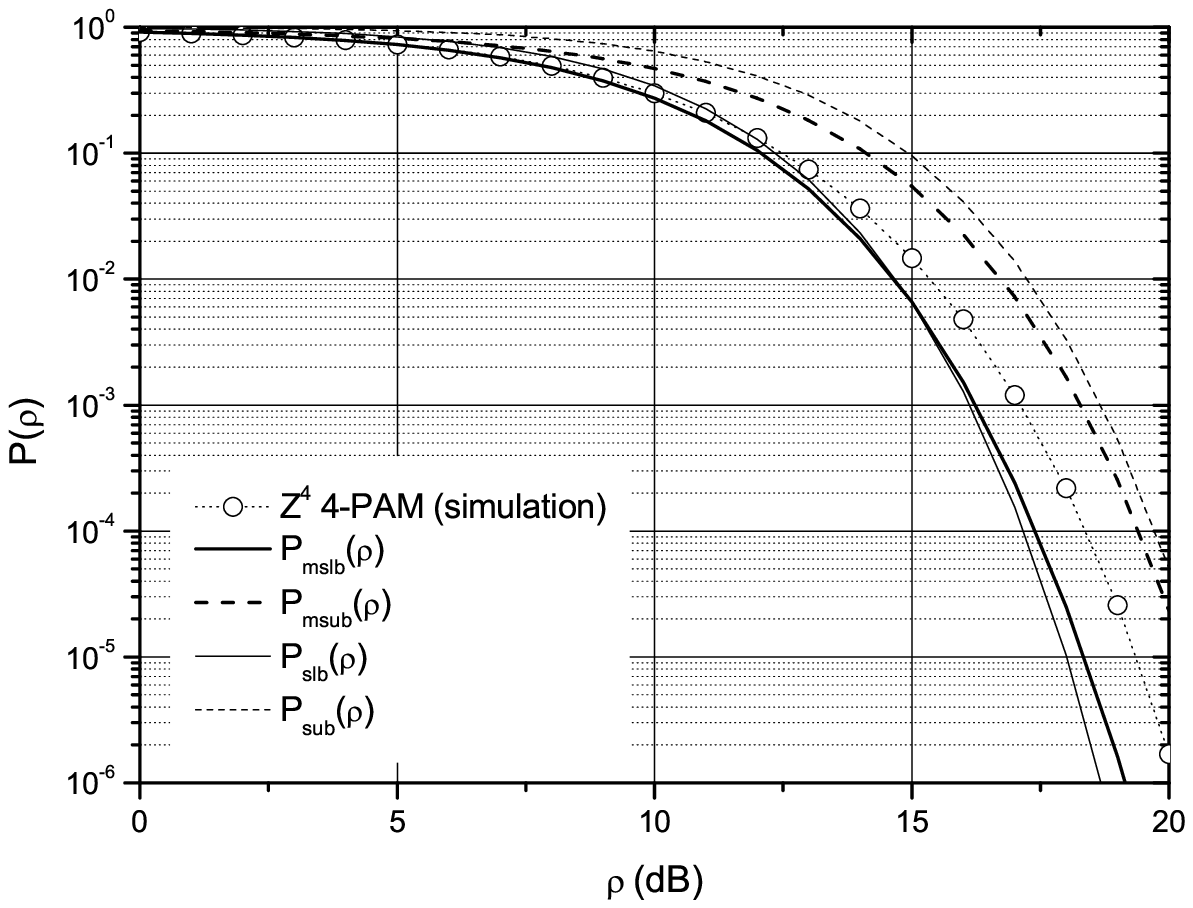}
\center\caption{Symbol Error Probability, MSLB and MSUB for the
$\mathbb{Z}^4$ $4-PAM$ constellation and SLB and SUB for the
$\mathbb{Z}^4$ lattice.} \label{Fig:z4-4pam}
\end{figure}

\begin{figure}[h!]
\centering\includegraphics[keepaspectratio,width=6in]{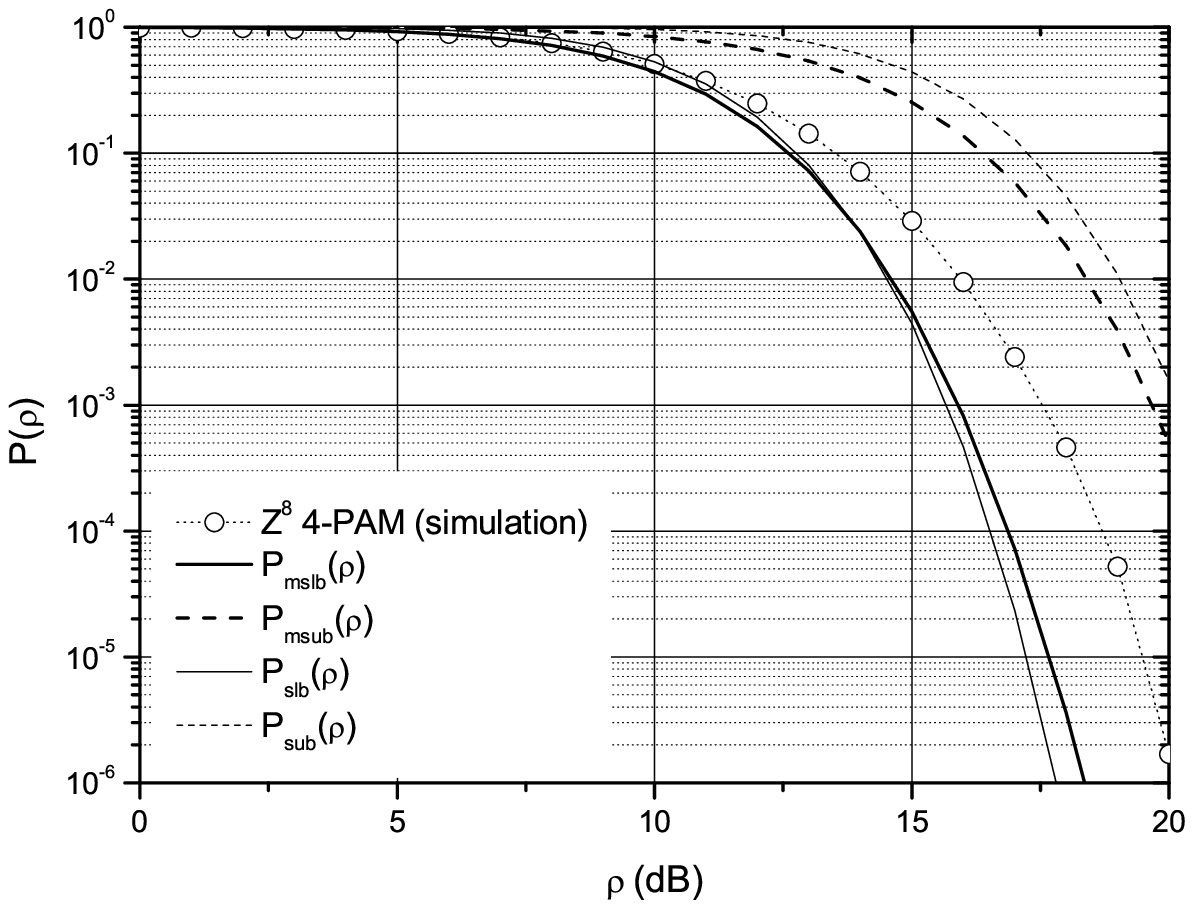}
\center\caption{Symbol Error Probability, MSLB and MSUB for the
$\mathbb{Z}^8$ $4-PAM$ constellation and SLB and SUB for the
$\mathbb{Z}^8$ lattice.} \label{Fig:z8-4pam}
\end{figure}

\begin{figure}[h!]
\centering\includegraphics[keepaspectratio,width=6in]{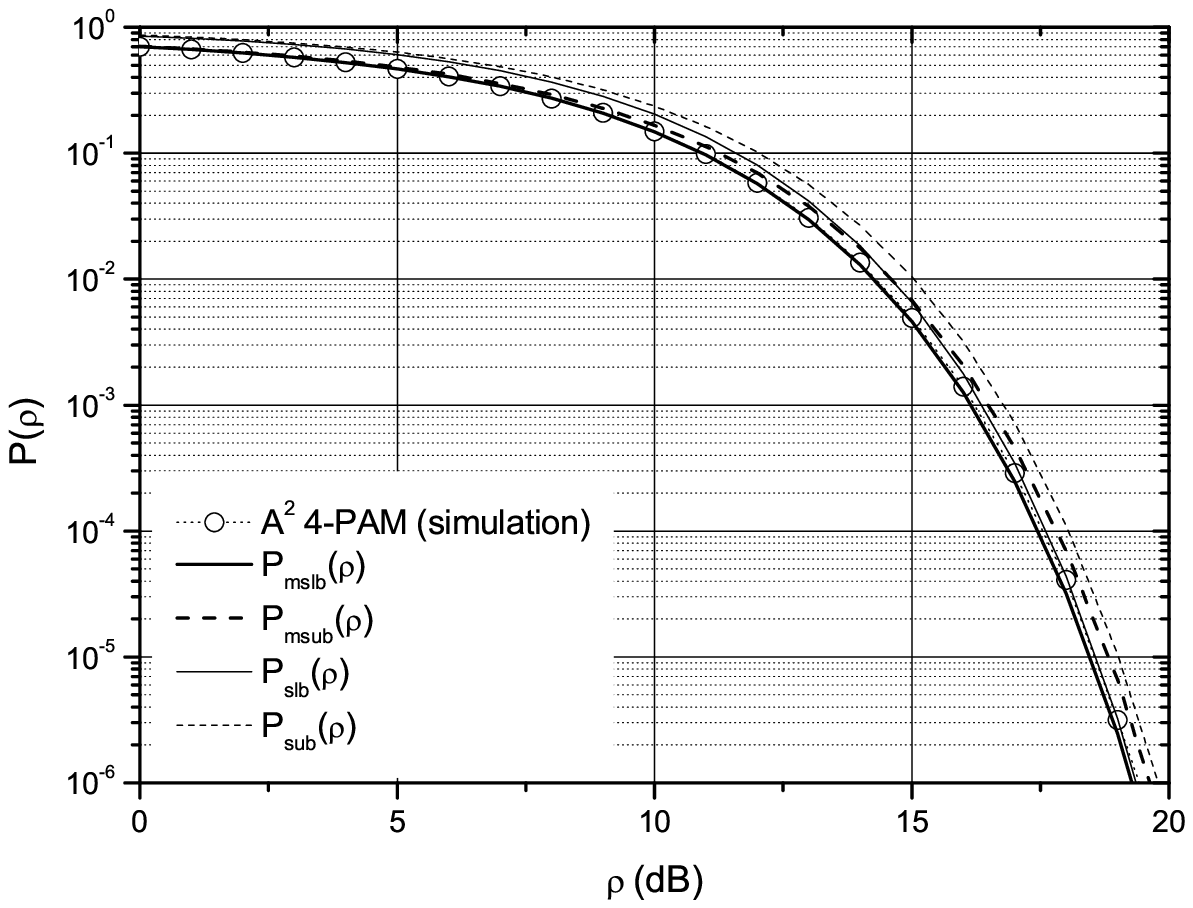}
\center\caption{Symbol Error Probability, MSLB and MSUB for the
$\mathbb{A}^2$ $4-PAM$ constellation and SLB and SUB for the
$\mathbb{A}^2$ lattice.} \label{Fig:a2-4pam}
\end{figure}

\begin{figure}[h!]
\centering\includegraphics[keepaspectratio,width=6in]{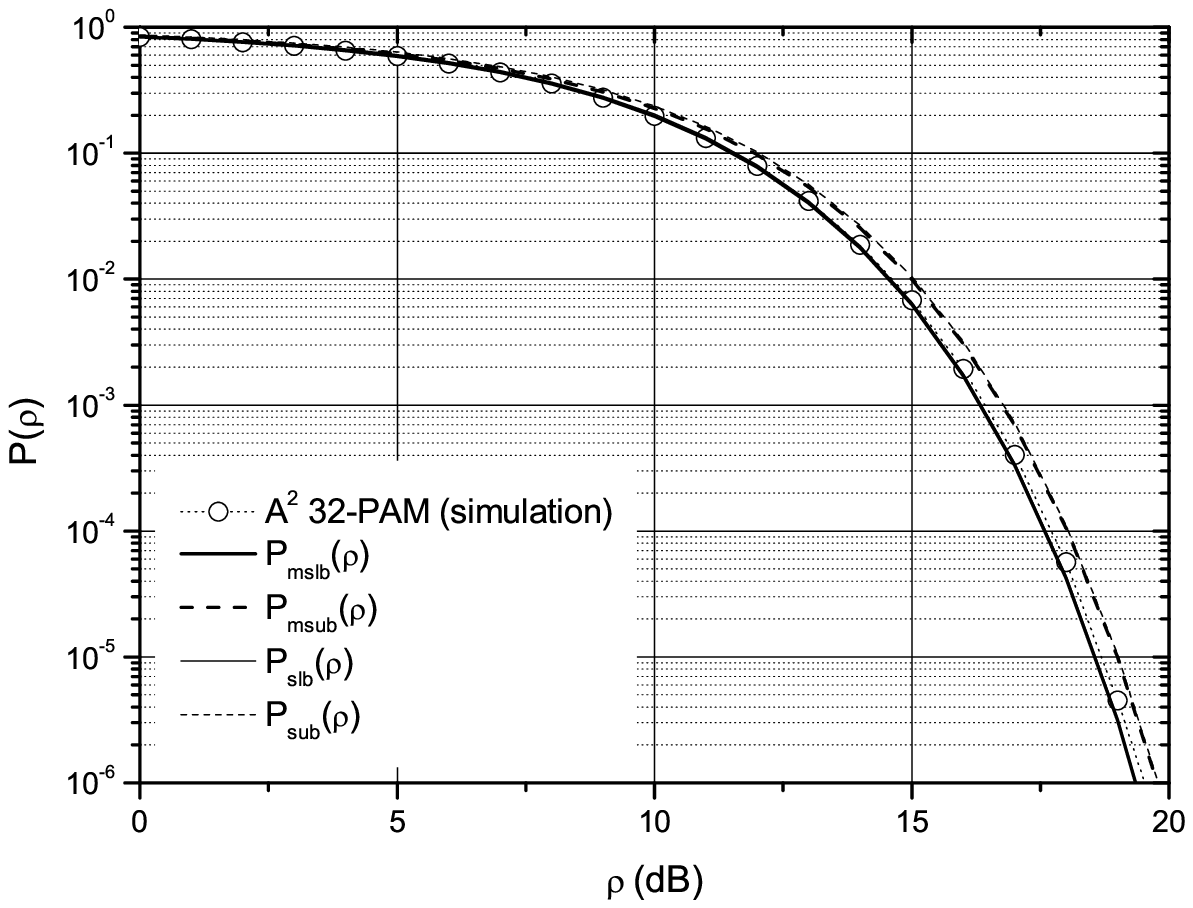}
\center\caption{Symbol Error Probability, MSLB and MSUB for the
$\mathbb{A}^2$ $32-PAM$ constellation and SLB and SUB for the
$\mathbb{A}^2$ lattice.} \label{Fig:a2-32pam}
\end{figure}

\begin{figure}[h!]
\centering\includegraphics[keepaspectratio,width=6in]{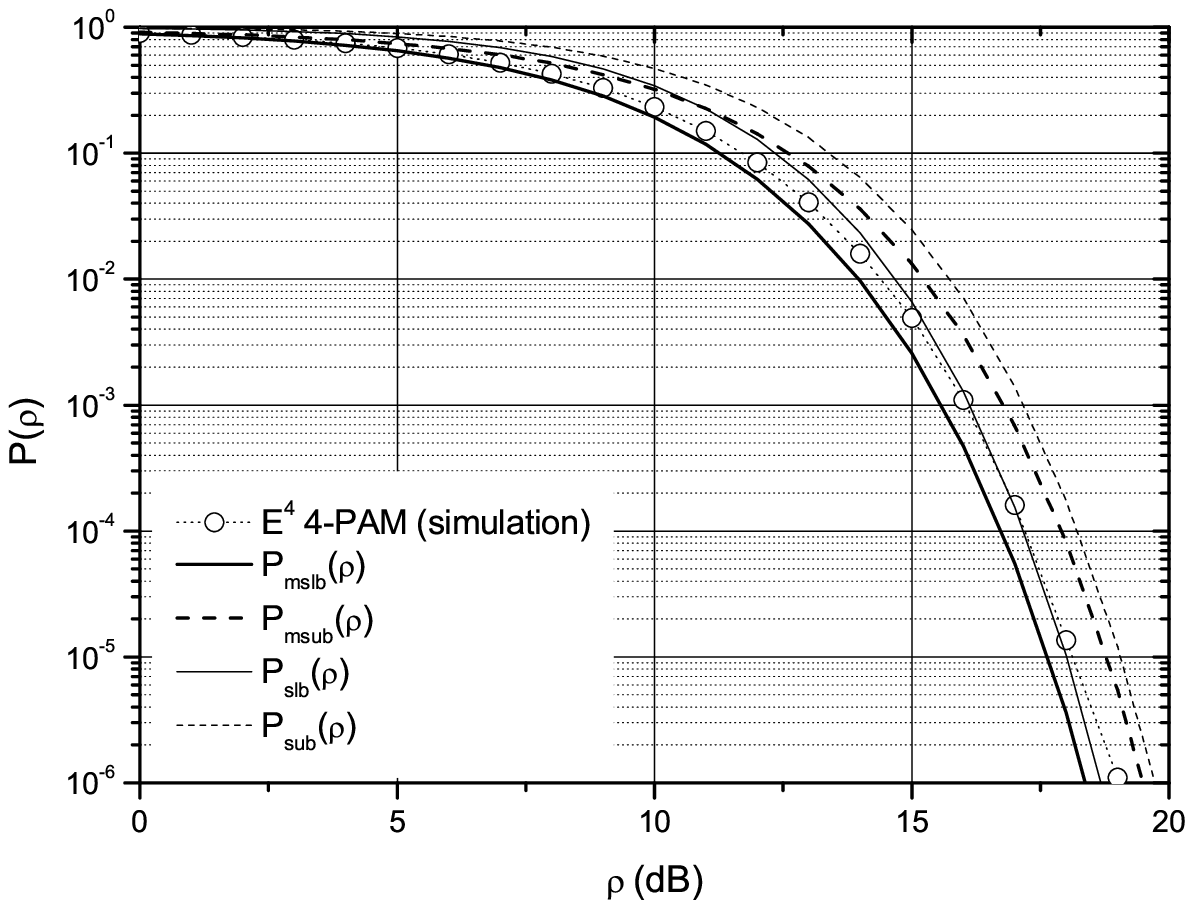}
\center\caption{Symbol Error Probability, MSLB and MSUB for the
$\mathbb{E}^4$ $4-PAM$ constellation and SLB and SUB for the
$\mathbb{E}^4$ lattice.} \label{Fig:e4-4pam}
\end{figure}

\begin{figure}[h!]
\centering\includegraphics[keepaspectratio,width=6in]{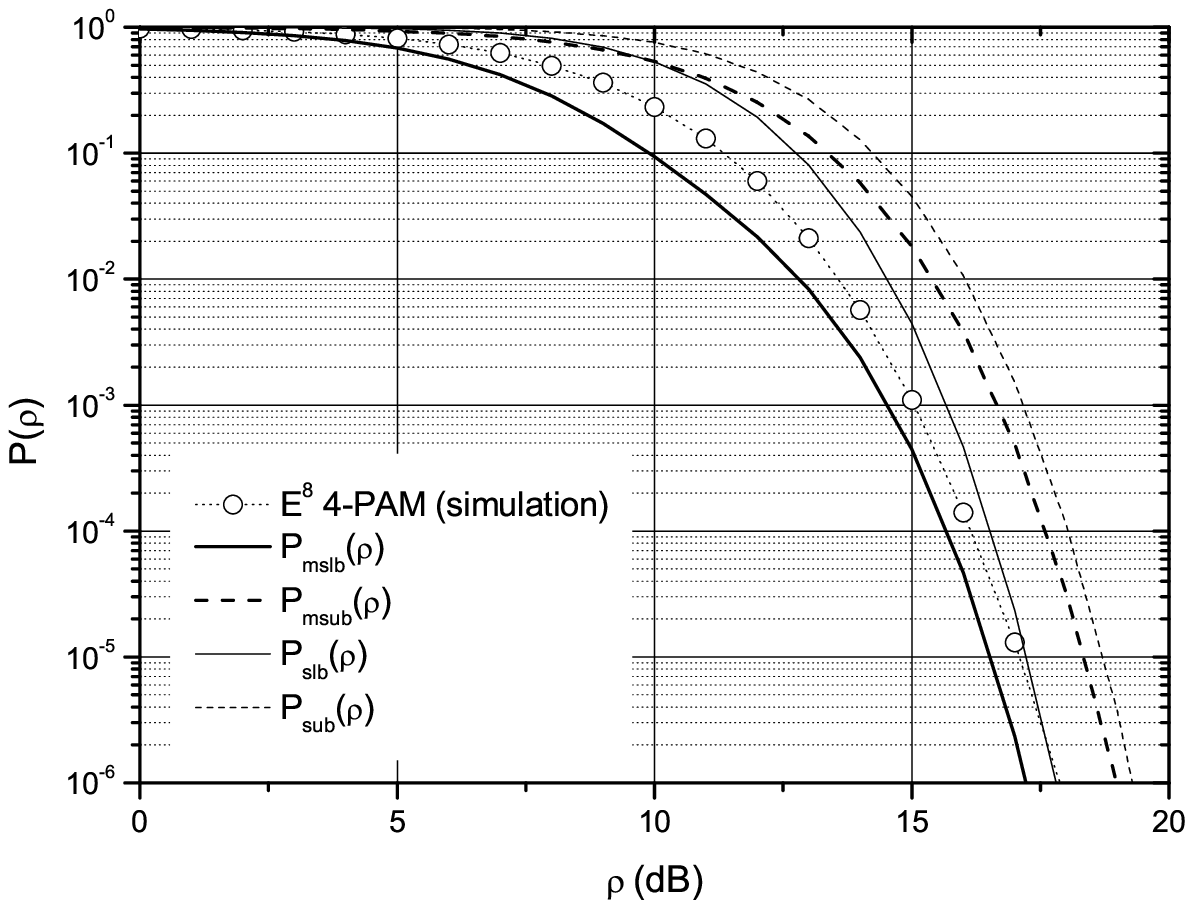}
\center\caption{Symbol Error Probability, MSLB and MSUB for the
$\mathbb{E}^8$ $4-PAM$ constellation and SLB and SUB for the
$\mathbb{E}^8$ lattice.} \label{Fig:e8-4pam}
\end{figure}

\end{document}